\documentclass[12pt,journal,draftcls,letterpaper,onecolumn]{IEEEtran}

\usepackage{cite}      

\usepackage{graphicx}  

\usepackage{psfrag}    

\usepackage{amsmath}   
\makeatletter
\let\IEEEproof\proof
\let\IEEEendproof\endproof
\let\proof\@undefined
\let\endproof\@undefined
\makeatother

\usepackage{amssymb,amsthm}

\let\proof\IEEEproof
\let\endproof\IEEEendproof

\usepackage{epsfig}
\usepackage{amsfonts}
\usepackage{mathrsfs}
\usepackage{euscript}
\usepackage{graphicx}

\newtheorem{lemma}{Lemma}

\theoremstyle{definition}

\begin{document}
%
\title{A Fixed Precoding Approach to Achieve the Degrees of Freedom in X channel}
\author{Soroush Akhlaghi, and Mohammad Ali Maddah-Ali~\IEEEmembership{Member,~IEEE}\\
Emails: akhlaghi@shahed.ac.ir, maddah-a@eecs.berkeley.edu\\
}

\maketitle

\begin{abstract}
This paper aims to provide a fixed precoding scheme to achieve the Degrees of Freedom~\emph{DoF} of the generalized ergodic X channel. This is achieved through using the notion of ergodic interference alignment technique. Accordingly, in the proposed method the transmitters do not require to know the full channel state information, while this assumption is the integral part of existing methods. Instead, a finite-rate feed-back channel is adequate to achieve the~\emph{DoF}. In other words, it is demonstrated that quantized versions of channel gains are adequate to achieve the~DOF. To get an insight regarding the functionality of the proposed method, first we rely on finite field channel models, and then extend the terminology to more realistic cases, including dispersive fading channels in the presence of quantizer. Accordingly, in a Rayliegh fading environment, it is shown a feedback rate of $2\log(p)+\theta(\log\log(p))$ can provide the~\emph{DoF}, where $p$ is the total transmit power. Moreover, the impact of low feedback rate on the multiplexing gain is investigated, where a formula between the achievable multiplexing gain and feedback rate is identified.
\end{abstract}


\section{Introduction}
%
%
%
%

\PARstart{T}{his} paper concerns communication over X channel with two transmitters and two receivers, in which each transmitter aims at sending an independent message to either of receivers. However, the material in this paper can be extended to a more general case with any arbitrary number of transmitters and receivers. X channel is regarded as a basic block of a multi-port wireless network, since it encompasses a large variety of known channels. For instance, X channel subsumes broadcast, multiple access, and the interference channel. As such, any findings in X channel, by some marginal changes, may be extended to its derivatives.

Recent advances through the surge in the task of finding the capacity of wireless networks have given a new insight into the\emph{ Degrees of Freedom~(DoF)} in such networks~\cite{jafar-x_network}, an idea which is first proposed in MIMO channels~\cite{maddah_signaling_multicast&broadcast}, and then extended to single-antenna X-channel in fast fading environments~\cite{jafar-x_network,jafar-capacity&DoF_x_channel}. This enables to compare the throughput of a multiuser network to that of a single user, as if there is not any interfering co-channel user. Moreover, it provides a quantitative measure which enables to compare the impact of various strategies on the asymptotic throughput of a wireless network at high Signal to Noise Ratio~(SNR) regime. Recently, the~\emph{DoF} of a multiuser X channel when the channel gains vary across time~(or frequency), is obtained in~\cite{jafar-capacity&DoF_x_channel,jafar-Interference_deterministic,Jafar-Interference&DoF_K_intchannel}. This is accomplished through using the notion of interference alignment technique, which basically concerns steering multiple interferes so that the desired signal can be easily distinguished from the unwanted signals at the corresponding receiver. This is in accordance to what is previously applied to the case of multiple antennas in~\cite{maddah_signaling_multicast&broadcast} to spatially align the non-intended signals~(interferers) in the same direction at the receiver part, thereby releasing more dimensions for the intended signal.

The same method can be easily applied for the single antenna X channel by exploiting different time slots/frequency bands, instead of spatial dimensions. For instance, the method proposed in~\cite{Jafar-Interference&DoF_K_intchannel} is proved to be useful in fast fading environment, for which different time slots mimic the required dimensions. Also, in~\cite{Motahari-real_numbers},~\cite{Motahari-real_single_antena} inspired by the notion of Diophantine approximation in number theory, an elegant method, dubbed \emph{Real Interference Alignment}, is deduced and is shown can achieve the same \emph{DoF} over time invariant channels, indicating slow fading channels do not fall short of achieving~\emph{DoF}.

Most of current works assume the CSI is perfectly available at the transmitters. However, this is not a realistic case happening in practical systems. Thus, it is desirable to investigate the achievable~\emph{DoF} in such channel when the partial CSI is causally available through a finite rate feedback channel.

To address the aforementioned issue, inspired by the pioneering works in~\cite{Bobak-Ergodic,Bobak-Interference_finite_SNR}, we propose using a \emph{fixed precoding approach} for the X channel in which the notion of ergodic interference alignment is employed. Moreover, to get more insight regarding the proposed method, first we rely on a finite field model, and then extend the terminology to a more realistic case, where the channel gains are continues random variables drawn from a Gaussian distribution. Accordingly, an elegant feed-forward strategy is deduced, showing the proposed approach still can achieve a DoF in such channel as long as the CSI is partially available at the transmitters. To this end, an elegant vector quantization method is proposed, showing one can approach the DoF as long as the quantization levels fall below a certain threshold.

Moreover, it is shown the feedback rate of $2\log(p)+\Theta(\log\log(p))$ is adequate to approach the~\emph{DoF}, where $p$ denotes the total transmit power. Finally, the impact of quantization errors on the achievable~\emph{DoF} is studied, where a formula between the achievable DoF and the required feedback is identified. In sum, the main contributions of the current work can be summarized in two parts as follows:

\begin{itemize}
  \item The introduction of a \emph{fixed precoding approach} to achieve the \emph{DoF}.
  \item To relates the finite-rate feedback link to the achievable~\emph{DoF} in X channel.
\end{itemize}

The rest of this paper is organized as follows. Next,
Section~\ref{sec:model} provides the system model. Then, Section~\ref{sec:history} presents the related research area and addresses shortcomings.
Section~\ref{sec:motivation_of_proposed} motivates the proposed method and has followed with Section~\ref{sec:finite} that introduce the finite field model. \ref{sec:analysis_of_proposed} formulates the proposed method and \ref{sec:Quantrate} investigate the impact of low quantization rate.
Finally conclusions and future works wrap up the paper.

Throughout the paper, boldface letters indicate vectors (lower case) or matrices (upper case). The $^{\dagger}$ notation denotes the conjugate transpose of a vector or a matrix. In addition $\angle(\bold x,\bold y)$ indicate the angle between two vectors $\bold x$ and $\bold y$ and the cosine of this angle is the inner product of the normalized vectors along the vectors $\bold x$ and $\bold y$. The null space of a vector $\bold{v}$ is denoted by N($\bold{v}$).

\section{System Model}\label{sec:model}
We consider a simple X channel composed of two transmitters and two receivers in which each transmitter aims at sending independent messages to either of receivers
(see Figure~\ref{fig:1}). Assuming the channel between the $i^{th}$ transmitter and the $j^{th}$ receiver at time instant $t$\footnote{Note that $t$ can be replaced by specific dimension, i.e., temporal, frequency or spatial dimensions.} is represented by $h_{ji}(t)$, the received signal by the $j^{th}$ receiver, $y_j(t)$ for $j=1,2$, is given by,
\begin{eqnarray}\label{equ1}
y_{1}(t)&=&h_{11}(t)x_1(t)+h_{12}(t)x_2(t)+n_1(t)\nonumber\\
y_{2}(t)&=&h_{21}(t)x_1(t)+h_{22}(t)x_2(t)+n_2(t)~.
\end{eqnarray}
In some cases, it is desirable to consider $M$ different time slots which are perceived to be complementary matched (the notion of complementary matched time slots will be discussed in Section~\ref{sec:analysis_of_proposed}). In what follows, the set of indexes of complementary matched time slots is called complementary set. Hence, one can readily rewrite~(\ref{equ1})~ for the $\tau^{th}$ complementary set, encompassing $M$ complementary time slots, using the following matrix notation,
\begin{eqnarray}\label{equ2}
\bold{y}_{1}(\tau)&=&\bold{H}_{11}(\tau)\bold{x}_1(\tau)+\bold{H}_{12}(\tau)\bold{x}_2(\tau)+\bold{n}_1(\tau)\nonumber\\
\bold{y}_{2}(\tau)&=&\bold{H}_{21}(\tau)\bold{x}_1(\tau)+\bold{H}_{22}(\tau)\bold{x}_2(\tau)+\bold{n}_2(\tau)~.
\end{eqnarray}
Assuming $\{{\pi_1},{\pi_2},\ldots,{\pi_M}\}$ represents the corresponding channel usage index for the $\tau^{th}$ complementary set, it follows
\begin{eqnarray}\label{equ3}
\bold{x}_{i}(\tau)\!\!\!\!&=&\!\!\!\![x_i(\pi_1),x_i(\pi_2),\ldots,x_i(\pi_M)]^T~\textrm{for}\;i=1,2\nonumber\\\nonumber\\
\bold{H}_{ij}(\tau)\!\!\!\!&=&\!\!\!\!
\left(
\begin{array}{cccc}
h_{ij}(\pi_1) & 0 & \ldots & 0 \\
0 & h_{ij}(\pi_2) & \ldots & 0 \\
\vdots & \vdots & \ddots & \vdots \\
0 & 0 & \ldots & h_{ij}(\pi_n)
\end{array}
\right)\nonumber\\
&~&\quad\quad\quad\quad\quad\quad\quad\quad\textrm{for}\;i,j=1,2
\end{eqnarray}
Also, it is assumed the $i^{th}$ transmitter sends a data stream $d_{ji}$ to the $j^{th}$ receiver, for $i,j=1,2$. To this end, each transmitter may send either of data streams along \emph{distinct} directions, dubbed beamforming directions, i.e.,
\begin{eqnarray}\label{equ4}
\bold{x}_{1}(\tau)\!\!\!\!&=&\!\!\!\!d_{11}\bold{v}_{11}(\tau)+d_{21}\bold{v}_{21}(\tau)\nonumber\\
\bold{x}_{2}(\tau)\!\!\!\!&=&\!\!\!\!d_{12}\bold{v}_{12}(\tau)+d_{22}\bold{v}_{22}(\tau)
\end{eqnarray}
where in~(\ref{equ4}), $\bold{v}_{ji}$ for $i,j=1,2$, denotes the assigned direction corresponding to the data stream $d_{ji}$. Thus, the received signal at the $j^{th}$ receiver, for $j=1,2$, would be,
\begin{eqnarray}\label{equ5}
\bold{y}_{1}(\tau)\!\!\!\!\!\!&=&\!\!\!\!\!\!d_{11}\bold{H}_{11}(\tau)\bold{v}_{11}(\tau)+d_{12}\bold{H}_{12}(\tau)\bold{v}_{12}(\tau)+\bold{i}_1(\tau)+\bold{n}_1(\tau)\nonumber\\
\bold{y}_{2}(\tau)\!\!\!\!\!\!&=&\!\!\!\!\!\!d_{21}\bold{H}_{21}(\tau)\bold{v}_{21}(\tau)+d_{22}\bold{H}_{22}(\tau)\bold{v}_{22}(\tau)+\bold{i}_2(\tau)+\bold{n}_2(\tau)\nonumber\\
\end{eqnarray}
where $\bold{i}_j$ for $j=1,2$ are defined as,
\begin{eqnarray}\label{equ6}
\bold{i}_{1}(\tau)\!\!\!\!&=&\!\!\!\!d_{21}\bold{H}_{11}(\tau)\bold{v}_{21}(\tau)+d_{22}\bold{H}_{12}(\tau)\bold{v}_{22}(\tau)\nonumber\\
\bold{i}_{2}(\tau)\!\!\!\!&=&\!\!\!\!d_{11}\bold{H}_{21}(\tau)\bold{v}_{11}(\tau)+d_{12}\bold{H}_{22}(\tau)\bold{v}_{12}(\tau)~.
\end{eqnarray}
It is worth mentioning that the terms $\bold{i}_j$ for $j=1,2$ can be treated as interference terms, as they entail
non-intended signals for the corresponding receiver.
\begin{figure}[htp]
\centering
\epsfig{file=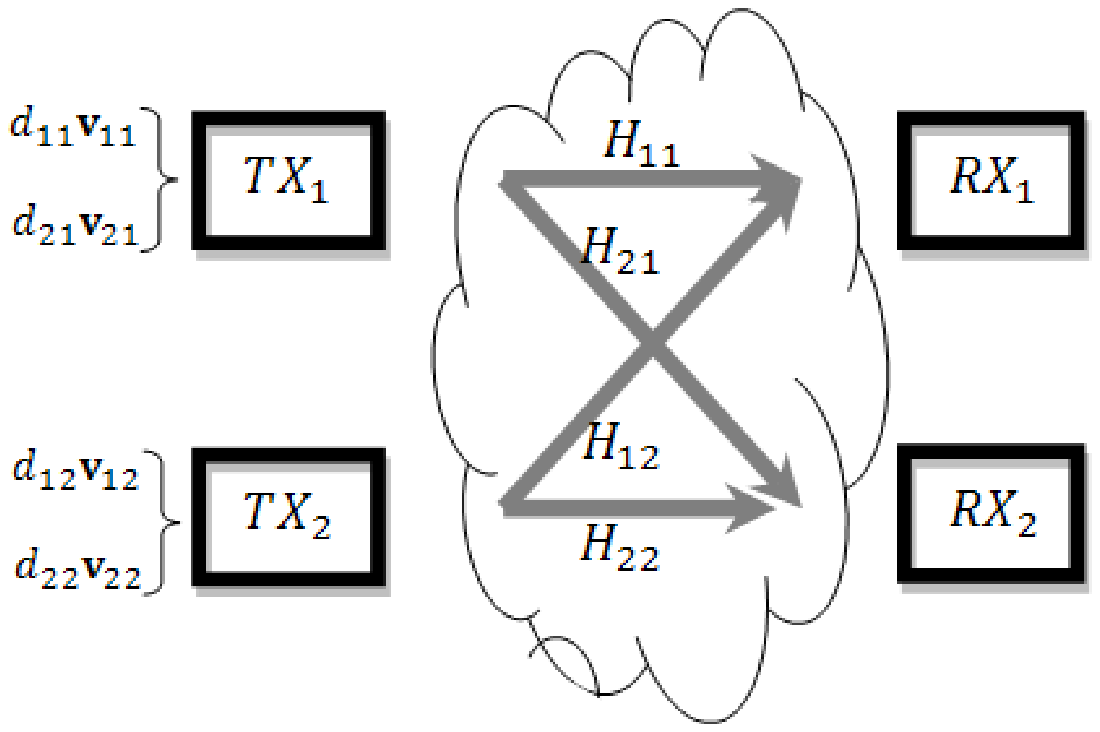,width=0.65\linewidth,clip=1}
\caption{The block diagram of MIMO X channel.}
\label{fig:1}
\end{figure} 

\section{Background Information}\label{sec:history}
This section aims to address recent advances through exploring the ~\emph{DoF} of X channels. It should be noted that the ~\emph{DoF} in a network with finite number of users is defined as the ratio of sum-capacity,~$C(\rho)$, over the signal to noise ratio, $\rho$, in log scale as the signal to noise ratio tends to infinity~\cite{Poon-DoF,Host-MG_network,Jafar-DoF_mimo_intchannel}, i.e.,
\begin{eqnarray}\label{equ7}
DoF=\lim_{\rho \to \infty} \frac{C(\rho)}{\log(\rho)}
\end{eqnarray}
The ~\emph{DoF} of X channel is first explored in~\cite{maddah_signaling_multicast&broadcast} for a $2\times 2$ MIMO X channel and then extended in~\cite{jafar-x_channel} to more general cases for an arbitrary number of transmitters/receivers and any number of antennas. In all of the aforementioned research works, the notion of interference alignment technique is employed, which basically concerns to steer the interfering signals in a small sub-space, thereby releasing more dimensions for the desired signals. As a result, assuming individual data streams are sent over $M$ time slots, the resulting ~\emph{DoF} is simply computed as the total number of interference-free dimensions over $M$~(all used dimensions). For instance, referring to~(\ref{equ6}), the interference term at the first receiver, $i_1(\tau)$, is composed of two interfering vectors: $d_{21}\bold{H}_{11}(\tau)\bold{v}_{21}(\tau)$ and $d_{22}\bold{H}_{12}(\tau)\bold{v}_{22}(\tau)$. Thus, it is desirable to devise beamforming vectors $\bold{v}_{21}$ and $\bold{v}_{22}$ such that the interfering vectors occupy the same direction at the receiver $1$. The same argument can be readily applied to the receiver $2$. As a result, it is demonstrated that the beamforming vectors $\bold{v}_{ij}$ for $i,j=1,2$ can be chosen such that the interfering vectors at each receiver lie in the same direction, i.e.,
\begin{eqnarray}\label{equ8}
\bold{H}_{11}(\tau)\bold{v}_{21}(\tau)&=&c_1\bold{H}_{12}(\tau)\bold{v}_{22}(\tau)\nonumber\\
\bold{H}_{21}(\tau)\bold{v}_{11}(\tau)&=&c_2\bold{H}_{22}(\tau)\bold{v}_{12}(\tau)
\end{eqnarray}
where $c_i$ for $i=1,2$, denote any arbitrary constant values. As a result, one direction at each receiver is reserved for the interference signals, and hence, this dimension can be readily eliminated by zero-forcing processing at each receiver and noting $M$ time slots are being used for transmission, hence, $M-1$ out of $M$ dimensions are retained~(the null-space of the interference signal) for the data streams to be sent to either of receivers. Since, either of receivers should receive totaly two independent data streams, each from either of transmitters, thus the null-space should be of rank $2$, or greater, i.e., $M-1\ge 2$. As a result, $M=3$ is the minimum number of dimensions for sending totaly $4$ non-interfering data streams (two data streams for each receiver), and hence, the resulting ~\emph{DoF} would be $\frac{4}{3}$~\cite{maddah_signaling_multicast&broadcast}. However, the interference alignment approach posses some impractical restrictions, i.e., it is assumed the CSI is perfectly available at the transmitters. This shortcoming is the primary source of motivations behind the current work. In the next section, motivated by the pioneering work in~\cite{Bobak-Ergodic,Bobak-Interference_finite_SNR}, a more realistic approach is proposed which addresses the aforementioned issue.

\section{Proposed method}\label{sec:motivation_of_proposed}
The interference alignment technique, as is noted earlier,
emerged out of the work on exploring the \emph{DoF} of MIMO
X channel~\cite{jafar-x_channel}, and then is identified as a promising approach
to discover the asymptotic capacity of a wide variety of wireless networks at high
SNR region, including the broadcast and the interference channels~\cite{maddah_signaling_multicast&broadcast,Jafar-Interference_K_intchannel}.
In two-user X channel, referring to the argument discussed in the preceding
section, the conventional interference alignment technique attempts to find
transmit beamforming vectors such that the resulting interfering signals at
each receiver occupies the same direction. Accordingly, in~\cite{Jafar-DoF_MIMO_X_channel} an
elegant approach is proposed to compute the beamforming vectors based on
the causal CSI at the transmitters. However, a forward link is required to
send the beamforming vectors to the receivers. Moreover, the transmitters need
to know the perfect CSI to determine the beamforming vectors.
However, for a broad variety of ergodic channels, another variation of
interference alignment is recently proposed, which is called the ergodic
interference alignment~\cite{Bobak-Ergodic}, an idea which is proved to draw a concrete path towards exploring the \emph{DoF} of more sophisticated networks. For instance, through using the notion of ergodic interference alignment, it is shown in a K-user interference
network, each user can achieve half of its interference-free ergodic
capacity~\cite{Bobak-Interference_finite_SNR}.
This is achieved through finding a set of channel indexes which form a
complementary set according to some criteria, and sending the same data
stream over these dimensions, so that the resulting channels between each
transmitter to the affiliated receiver seems as if the there is no
interference.
In this work, we generalize the concept of ergodic interference alignment
to X channel, and show it leads to a fixed-precoding approach for the
entire transmission.
In what follows, we first concentrate on the finite-field model for the channel gains to get an intuition how our proposed approach works and then extent to more general cases, including the fading channel.

\section{Ergodic Interference alignment for Finite Field two-user X Channel}\label{sec:finite}
We consider $\mathcal{A}$ as a set of $|\mathcal{A}|$ distinct elements, i.e., $\mathcal{A}=\{a_1,a_2,\ldots,a_{|\mathcal{A}|}\}$. We also assume $(\mathcal{A},+,\odot)$ is a Field, where $+$ and $\odot$ denote, respectively, the addition and multiplication signs. As a result, $(\mathcal{A},+)$ and $(\mathcal{A}-\{0\},\odot)$ form commutative groups, where $0$ is the identity element of addition sign. We also define a vector space $\mathcal{V}$ of dimension three, whose elements are three-tuple vectors with entries from the set $\mathcal{A}$. Also, the element-wise addition and scalar multiplications over the set $\mathcal{V}$ are assumed to be, respectively, in accordance to the addition and multiplication signs of the field $(\mathcal{A},+,\odot)$.

It can be shown that the vector space defined over $\mathcal{A}$ is a commutative group, thus associative and commutative laws are verified. The scalar product between a vector in the vectorial space, $\bold v \in \mathcal{V}$, and a scalar, $c\in \mathcal{A}$, dubbed scaler product, is defined as
\begin{displaymath}
c \bullet  \bold v = \left(\begin{array}{c}
                 c\odot v_1 \\
                 c\odot v_2 \\
                 c\odot v_3
               \end{array}\right)
\end{displaymath}
Additionally, two vectors $\bold{v}_1, \bold{v}_2\in \mathcal{V}$ have the same direction if there is a non-zero scalar $c\in \mathcal{A}$ for which the following equality holds,
\begin{displaymath}
\bold{v}_1 = c \bullet \bold{v}_2
\end{displaymath}

In what follows, we present the concept of ergodic interference alignment for finite field two-user X channel to get an indication regarding the proposed method, assuming the channel gains are uniformly chosen from the elements of $\mathcal{A}$ and independently vary across time. Moreover, it is simply assumed the channel is noise free as the objective is to see how the interference term is managed~(see Figure~\ref{fig:2} for the system model). As is mentioned in section~\ref{sec:model}, there are four data streams $d_{ij}$ for $i,j=1,2$ to be sent to the intended receivers. In this work, we simply assume precoding vectors are any arbitrary disjoint vectors chosen from vector space $\mathcal{V}$ defined over the set $\mathcal{A}$. Also, it is assumed the channel gains are randomly chosen from the entries of set $\mathcal{A}$. Moreover, we simply discard those time indexes for which the channel gain between a transmitter to a receiver is zero. Clearly, this has a modest impact on the \emph{DoF} when the cardinality of the set $\mathcal{A}$ is large enough. In what follows, we will describe a method which aims at identifying the complementary matched time slots. As is mentioned in the preceding section, we are going to classify time slots in triple sets in which (\ref{equ8}) holds for each set. To this end, the proposed algorithm can be summarized in the following steps,
\begin{figure*}[htp]
\centering
\epsfig{file=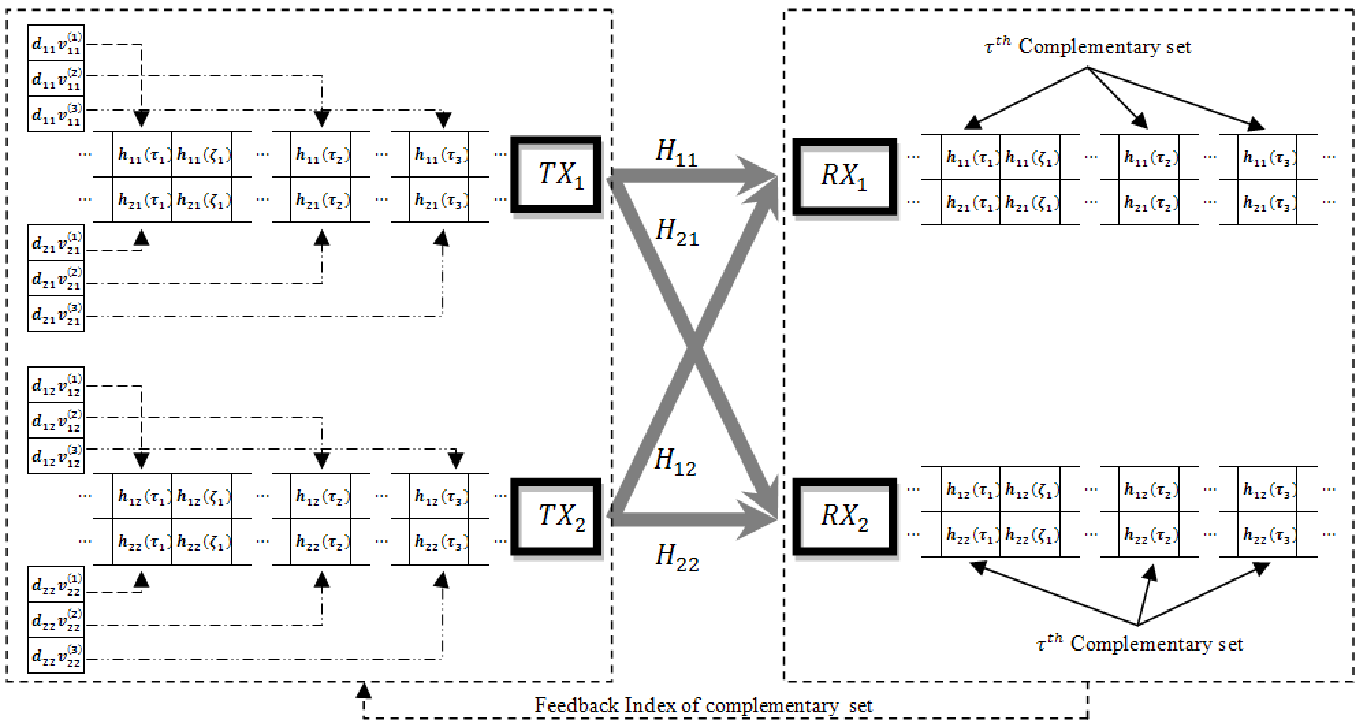,width=1\linewidth,clip=}
\caption{The block diagram of the proposed method for two-user X channel.}
\label{fig:2}
\end{figure*}
\begin{itemize}
  \item Step 1: Set $t=1$, $\tau=1$, and consider the $t^{th}$ time slot in the $\tau^{th}$ complementary set, $\Pi_\tau$, thus $1\in \Pi_1$.
  \item Step 2: Compute the constant terms, $c_i$ for $i=1,2$, defined in (\ref{equ8}) associated with the $\tau^{th}$ complementary set as $c_1=\frac{h_{11}(t)\odot \bold{v}_{21}^t}{h_{12}(t)\odot \bold{v}_{22}^t}\in \mathcal{A}$ and $c_2=\frac{h_{21}(t)\odot \bold{v}_{11}^t}{h_{22}(t)\odot \bold{v}_{12}^t}\in \mathcal{A}$, where $t$ is the first selected time slot in the corresponding complementary set, thus $c_1=\frac{h_{11}(1)\odot \bold{v}_{21}^1}{h_{12}(1)\odot \bold{v}_{22}^1}$ and $c_2=\frac{h_{21}(1)\odot \bold{v}_{11}^1}{h_{22}(1)\odot \bold{v}_{12}^1}$ for $\Pi_1$.
  \item Step 3: Set t=t+1, compute $c_1=\frac{h_{11}(t)\odot \bold{v}_{21}^t}{h_{12}(t)\odot \bold{v}_{22}^t}$ and $c_2=\frac{h_{21}(t)\odot \bold{v}_{11}^t}{h_{22}(t)\odot \bold{v}_{12}^t}$. Then verify if these values are equal to $c_1$ and $c_2$ associated with existing complementary sets, if there is any set with the same $c_1$ and $c_2$, put this time slot into this complementary set; otherwise set $\tau=\tau+1$ and initiate $\Pi_\tau$ as the $\tau$'th complementary set, then put time slot $t$ into $\Pi_\tau$ and assign the computed $c_1$ and $c_2$ to this set.
  \item Step 4: Check whether there is a complementary set among existing initiated sets of size $M=3$, if so, declare this set as a \emph{complementary matched set}, and proceeds the algorithm for further investigation to complete/initiate other \emph{complementary matched sets}.
  \item Step 5: Go to Step 3.
  \end{itemize}

Thus, having aware of complementary matched sets, the proposed interference alignment algorithm can be summarized in the following,

\begin{itemize}
  \item We assume $\bold{v}_{ij}$ for $i,j=1,2$ are randomly chosen beamforming vectors with entries from the set $\mathcal{A}$, and are assumed to be fixed for all transmissions. Moreover, it is assumed $D_{\Pi_\tau}=\{d_{11},d_{12},d_{21},d_{22}\}$ is the set of information signals to be sent over the $\tau$'th complementary matched set (see equation (\ref{equ4})).
  \item At each time instant, find this time instant belongs to which complementary set, i.e., $t\in\Pi_{\tau}$. Moreover find this time instant belongs to which position in this set, i.e., the $l$'th position where $l$ is an integer value between 1 and $M$ (here $M=3$). Then, referring to (\ref{equ4}), send $d_{1i}\bold{v}^{(l)}_{1i}+d_{2i}\bold{v}^{(l)}_{2i}$ from the $i$'th transmit antenna, where $d_{ij}\in D_{\Pi_{\tau}}$ and $\bold{v}^{(l)}$ are scalar values and $\bold{v}^{(l)}$ is denoting the $l$'th element of vector $\bold{v}$.
  \item At each receiver, form the received signals as M-tuples (here $M=3$) according to the complementary sets to make a received signal vector.
  \item As the received interference signals arising at the receivers are aligned, thus, one can simply project the received signal vector at each receiver to the null-space of the aligned interference vector and thus the desired signals \big($\{d_{11},d_{12}\}$ for receiver 1 and \{$d_{21},d_{22}\}$ for receiver 2\big) can be simply decoded, as if there is no interference.
  \end{itemize}
Note that at each time instant, the values $c_1=\frac{h_{11}(t)\odot \bold{v}_{21}^t}{h_{12}(t)\odot \bold{v}_{22}^t}$ and $c_2=\frac{h_{21}(t)\odot \bold{v}_{11}^t}{h_{22}(t)\odot \bold{v}_{12}^t}$ belong to the set $\mathcal{A}$. Assuming $h_{ij}$ for $i,j=1,2$ and $\bold{v}_{ij}^t$ for $i,j=1,2$ are uniformly distributed over the non-zero elements of the set $\mathcal{A}$, thus $c_1$ and $c_2$ take uniformly a non-zero element of the set $\mathcal{A}$. As a result, assuming the set $\mathcal{A}$ is of size $|\mathcal{A}|$, thus for a non-zero element of the set $\mathcal{A}$, i.e., $a_i\in\mathcal{A}$, it follows $\textrm{Pr}_{c_1}(x=a_i)=\textrm{Pr}_{c_1}(x=a_i)=\frac{1}{|\mathcal{A}|-1}$. Thus the probability that $M=3$ randomly chosen time slots have the same $c_1$ becomes $\textrm{p}_1=(|\mathcal{A}|-1)\frac{1}{(|\mathcal{A}|-1)^3}=\frac{1}{(|\mathcal{A}|-1)^2}$. Similarly, with probability $\textrm{p}_2=\frac{1}{(|\mathcal{A}|-1)^2}$ these time slots have the same $c_2$. Finally, the probability that these time slots are complementary matched becomes $\textrm{p}_1\textrm{p}_2=\frac{1}{(|\mathcal{A}|-1)^4}$. Finally, the probability that one can find $M=3$ out of $n$ time slots which are proportionally matched becomes $\left(\begin{array}{c}
         n \\
         3
       \end{array}
\right)\textrm{p}_1\textrm{p}_2\approx \frac{n^3}{(|\mathcal{A}|-1)^4}$. After some manipulations and for large values of $|\mathcal{A}|$, in order to have this probability approaching one, one can verify that $n$ scales as $\theta(|\mathcal{A}|^{\frac{4}{3}})$. This intuitively reflects the delay which is imposed to the system by the use of the proposed method.

\section{Analysis of Proposed Method}\label{sec:analysis_of_proposed}
In the preceding section, to get an insight regarding the proposed approach, the main idea is thoroughly discussed for the finite field model. In this section, we aim to extend this terminology to a more realistic case. To this end, it is assumed there is a common finite rate feedback channel to the transmitters, which merely provides some information regarding the complementary set. More precisely, the feedback channel defines the current time slot belongs to which complementary set. Fig.\ref{fig:2} illustrates the block diagram of the proposed method. As a result, each transmitter sends the same information over time slots of the same complementary set to virtually make a MIMO channel. Moreover, it is assumed the beamforming vectors, $\bold{v}_{ij}$ for $i,j=1,2$, are randomly chosen vectors. In what follows, we will introduce the proposed algorithm which aims at finding the complementary set.s

It is assumed the CSI is perfectly available at the receivers, and each receiver makes use of Random Vector Quantization (RVQ)~\cite{Santipach-limited_feedback,Santipach}. Moreover, the random vector quantizer $\mathcal{C}$ contains $2^{B}$ isotropic random vectors, which are uniformly distributed on the M-dimensional unit sphere. Accordingly, for a given M-dimensional vector, $\bold{q}$, the $i^{th}$ vector in the set $\mathcal{C}$ is chosen as the corresponding quantization vector, provided that the following condition holds
{\setlength\arraycolsep{1pt}
\begin{eqnarray}\label{equ9}
i&=&~\textrm{arg}~\max_{j=1,\ldots,2^{B}} |\tilde{\bold{q}}^{\dagger}\bold{w}_{j}|^2\nonumber\\
~&=&~\textrm{arg}~\max_{j=1,\ldots,2^{B}} \cos^{2}(\angle(\tilde{\bold{q}},\bold{w}_{j}))\nonumber\\
~&=&~\textrm{arg}~\min_{j=1,\ldots,2^{B}} \sin^{2}(\angle(\tilde{\bold{q}},\bold{w}_{j})),
\end{eqnarray}}
where in (\ref{equ9}), it is assumed $\tilde{\bold{q}}=\frac{\bold{q}}{|\bold{q}|}$, meaning the direction of $\bold{q}$ is merely quantized. Moreover, $\bold{w}_j$ denotes the $j^{th}$ vector in the set $\mathcal{C}$. Now, we are ready to describe the proposed algorithm. Referring to (\ref{equ8}), one can select the set of time indexes in which the selected random vectors (out of existing $2^B$ vectors) corresponding to the received vectors $\bold{H}_{11}\bold{v}_{21}$ and $\bold{H}_{21}\bold{v}_{11}$ are, respectively, the same as that of the vectors $\bold{H}_{12}\bold{v}_{22}$ and $\bold{H}_{22}\bold{v}_{12}$. In other words, the complementary matched time slots are selected such that the quantized direction of vector $\bold{H}_{11}\bold{v}_{21}$ is the same as that of $\bold{H}_{12}\bold{v}_{22}$, and this concurrently happens for the vectors $\bold{H}_{21}\bold{v}_{11}$ and $\bold{H}_{22}\bold{v}_{12}$. Once these time slots are identified, one can readily inform transmitters through the common finite-rate feedback channel to send the same information over these time slots. From now on, we call the aforementioned procedure as the proposed alignment method.

It is worth mentioning that for randomly chosen $M$ time slots, the probability that the resulting quantization vectors corresponding to the received vectors $\bold{H}_{11}\bold{v}_{21}$ and $\bold{H}_{12}\bold{v}_{22}$ to be the same is $2^{-B}$. Noting with the same probability these time slots are complementary matched for the second receiver, thus the probability that one can find a set of M time slots out of existing n time slots for which the interfering vectors at both receivers are aligned is $\left(\begin{array}{c}
         n \\
         M
       \end{array}
\right)2^{-B}2^{-B}\approx \frac{n^M2^{-2B}}{M!}$ for large $n$. Finally, the maximum value of $B$ for which this probability approaches one, scales as $B_{max}=\frac{M}{2}\log(n)-\Theta(\log\log(n))$. This means, as long as $B$ falls below $B_{max}$, with probability approaching one, there exist $M$ out of $n$ time slots for which $\bold{H}_{11}\bold{v}_{21}$ and $\bold{H}_{21}\bold{v}_{11}$ are steered along $\bold{H}_{12}\bold{v}_{22}$ and $\bold{H}_{22}\bold{v}_{12}$, respectively. It should be noted that RVQ may not be the optimum quantization method, however, it is demonstrated that this approach asymptotically approaches the optimum quantization method for the large value of $B$ \cite{Jindal-finite_feedback}. Moreover, as is shown later RVQ lends itself to a simplified mathematical formulation for the underlaying problem.

For the reminder of this section, we turn our attention to the impact of $B$ on the resulting \emph{DoF} in X channel. Assuming the $\tau^{th}$ complementary set is identified based on the proposed algorithm defined earlier, for the reminder of this section we simply discard the complementary set number $\tau$ to simplify the notations. Moreover, due to the symmetrical properties, we will restrict our attention to the receiver one. As a result, referring to~(\ref{equ5}) and noting above, the received signal vector for the $\tau^{th}$ complementary set of $M$ time slots at the first receiver can be rewritten as
\begin{eqnarray}\label{equ10}
\bold{y}_{1}=d_{11}\bold{q}_{11}+d_{12}\bold{q}_{12}+\bold{i}_1+\bold{n}_1
\end{eqnarray}
where in (\ref{equ10}), it is assumed $\bold{q}_{11}=\bold{H}_{11}\bold{v}_{11}$ and $\bold{q}_{12}= \bold{H}_{12}\bold{v}_{12}$. Also, the interfering vector $\bold{i}_1$ is
\begin{eqnarray}\label{equ11}
\bold{i}_{1}=d_{21}\bold{q}_{21}+d_{22}\bold{q}_{22}~.
\end{eqnarray}
Again, it is assumed $\bold{q}_{21}=\bold{H}_{11}\bold{v}_{21}$ and $\bold{q}_{22}=\bold{H}_{12}\bold{v}_{22}$.
Recall that, referring to the proposed alignment method, we have
\begin{eqnarray}\label{equ12}
i=\textrm{arg} \max_{j=1,\ldots,2^{B}} |\tilde{\bold{q}}^{\dagger}_{21}\bold{w}_{j}|^2=
\textrm{arg} \max_{j=1,\ldots,2^{B}} |\tilde{\bold{q}}^{\dagger}_{22}\bold{w}_{j}|^2
\end{eqnarray}
Assuming $\hat{\bold{q}}_{ij}$ for $i,j=1,2$ denote the quantization vectors corresponding to $\tilde{\bold{q}}_{ij}$ for $i,j=1,2$, thus referring to (\ref{equ12}), it follows $\hat{\bold{q}}_{21}=\hat{\bold{q}}_{22}=\bold{w}_{i}$ ( note that $\bold{w}_j$ for $j=1,\ldots,2^B$ are of unit norm). Also, we define\footnote{One could also replace $\tilde{\bold{q}}_{21}$ with $\tilde{\bold{q}}_{22}$ in the definition of $z_j$, as we are just interested in finding the distribution of $z_{max}$.} $z_j=|\tilde{\bold{q}}^{\dagger}_{21}\bold{w}_{j}|^2$ for $j=1,\ldots,2^B$. Thus, referring to equation (\ref{equ12}), we have $z_i=z_{max}$. Equation (\ref{equ12}) states that $z_i$ is the maximum of $2^{B}$ random variables, i.e., $z_i= \max_{j=1,\ldots,2^B} z_j$, where we have
\begin{eqnarray}\label{equ13}
z_{j}&=&|\tilde{\bold{q}}^{\dagger}_{21}\bold{w}_{j}|^{2}\nonumber\\
&=&\cos^{2}\left(\angle(\tilde{\bold{q}}_{21},\bold{w}_{j})\right)\nonumber\\
&=&1-\sin^{2}\left(\angle(\tilde{\bold{q}}_{21},\bold{w}_{j})\right),
\end{eqnarray}
As $\tilde{\bold{q}}_{21}$ and $\bold{w}_j$ are i.i.d. isotropic vectors with unit norm, thus $z_j$ is beta distributed with parameters $1$, $M-1$, i.e., $z_j\sim\beta(1,M-1)$, thus it has the following cumulative distribution function (cdf) ~\cite{Au-Yeung,Love-Grassmannian,Mukkavilli-beamforming_finite_feedback}:
\begin{eqnarray}\label{equ14}
F_{Z_{j}}(z_{j})=1-(1-z_{j})^{M-1},
\end{eqnarray}
Where $M$ is the dimension of vectors $\tilde{\bold{q}}_{12}$ and $\bold{w}_j$. As a result, the distribution of the maximum ($z_{max}$) of $2^B$ random variables, i.e., $z_j$ for $j=1,\ldots,2^B$ each drawn from a beta distribution with parameters $1$ and $M-1$, can be computed as,
\begin{eqnarray}\label{equ15}
F_{Z_{max}}(z_{max})&=&Pr(Z_{max}\leq z_{max})\nonumber\\
&=&Pr(Z_1\leq z_{max},\ldots,Z_{2^B}\leq z_{max})\nonumber\\
&=&Pr(Z_1\leq z_{max})^{2^B}\nonumber\\
&=&(1-(1-z_{max})^{M-1})^{2^{B}}~.
\end{eqnarray}
Note that we are interested in the distribution of $x\triangleq\sin^{2}(\angle(\tilde{\bold{q}}_{21},\hat{\bold{q}}_{21})=1-z_{max}$, which can be readily computed as,
\begin{eqnarray}\label{equ16}
\overline{F}_{X}(x)=(1-x^{M-1})^{2^{B}}~,
\end{eqnarray}
where $\overline{F}$ states the complementary of F. In the following, it is argued that the expectation of $x$ is strictly upper bounded as $2^{-\frac{B}{M-1}}$.

\begin{lemma}\label{lemma:upper-bounded}
The expectation of quantization error for RVQ technique can be upper bounded as,
\begin{displaymath}
E[\sin^{2}(\angle(\tilde{\bold{q}}_{kl},\hat{\bold{q}}_{kl})]~<~2^{-\frac{B}{M-1}}.
\end{displaymath}
\end{lemma}

\begin{proof}
for a complete proof refer to~\cite{Jindal-finite_feedback}.
\end{proof}
Referring to (\ref{equ10}) and (\ref{equ11}), and noting $d_{11}$ and $d_{12}$ are intended signals for the receiver one, the problem is to deduce a receive beamforming strategy to effectively remove the interfering signals, as if there is no interference. Assuming the interfering vectors in (\ref{equ11}), i.e., $d_{21}\bold{q}_{21}$ and $d_{22}\bold{q}_{22}$, are perfectly aligned, one could readily project the received signal in the null space of the aforementioned vectors, thereby canceling out the interference term in the null space. Moreover, in order to decode $d_{11}$ and $d_{12}$ in this null space, they should occupy different directions, thus the rank of null space should be at least two. As a result, we need at least three dimensions; one for the aligned interfering signals, and two directions for either of information signals $d_{11}$ and $d_{12}$. Consequently, we need to have $M=3$ and noting totaly $4$ information signals are sent to the respected receivers ($d_{ij}$ for $i,j=1,2$), thus the DOF becomes $\frac{4}{3}$.

However, this is not the case happening when using RVQ technique, as the interfering vectors are not perfectly aligned. In order to address the aforementioned issue, as the intended signals to receiver one are $d_{11}$ and $d_{12}$, we simply restrict our attention to decode $d_{11}$, as any findings, due to existing symmetrical properties for $d_{11}$ and $d_{12}$, can be readily extended to $d_{12}$. In what follows, the aforementioned issue is thoroughly discussed in more details.

In order to alleviate the impact of interfering signal vectors, one should project the received signal in the null space of the direction which is perceived to reduce the interference power (note that according to the proposed alignment strategy the interfering signals are not perfectly aligned along a single direction). In this work, referring to (\ref{equ12}), we set this direction along $\hat{\bold{q}}_{21}$ (or $\hat{\bold{q}}_{22}$) which is thought to be the best direction out of existing $2^{B}$ directions from the quantization vectors in the set $\mathcal{C}$. To this end, one need to multiply the received signal by $\bold{\Phi}_{21}$ (or equivalently $\bold{\Phi}_{22}$) which makes the projection space, thereby decreasing the interference power. Thus, we have,
\begin{eqnarray}\label{equ18}
\bold{\Phi}_{21}~=~\bold{I}-\hat{\bold{q}}_{21}\hat{\bold{q}}^{\dagger}_{21}.
\end{eqnarray}
As a result, applying $\bold{\Phi}_{21}$ to (\ref{equ10}) and noting (\ref{equ11}), it follows
{\setlength\arraycolsep{1pt}
\begin{eqnarray}\label{equ19}
\bold{\Phi}_{21}\bold{y}_{1}&=&d_{11}|\bold{q}_{11}|\bold{\Phi}_{21}\tilde{\bold{q}}_{11}+d_{12}|\bold{q}_{12}|\bold{\Phi}_{21}\tilde{\bold{q}}_{12}\nonumber\\
&+&d_{21}|\bold{q}_{21}|\bold{\Phi}_{21}\tilde{\bold{q}}_{21}+d_{22}|\bold{q}_{22}|\bold{\Phi}_{21}\tilde{\bold{q}}_{22}\nonumber\\
&+&\bold{\Phi}_{21}\bold{n}_1~.
\end{eqnarray}}
Recall that $\tilde{\bold{q}}_{ij}\triangleq\frac{\bold{q}_{ij}}{|\bold{q}_{ij}|}$ specifies the direction of vector $\bold{q}_{ij}$ and is of unit norm.
Moreover, $\bold{\Phi}_{21}\bold{q}_{ij}=|\bold{q}_{ij}|\bold{\Phi}_{21}\tilde{\bold{q}}_{ij}$ for $i,j=1,2$ is the projection of $\bold{q}_{ij}$ in the null space of $\hat{\bold{q}}_{21}$, thereby having dimension of size $M-1$. In what follows, for ease of notation, we define $\bold{q'}_{ij}\triangleq\bold{\Phi}_{21}\tilde{\bold{q}}_{ij}$ for $i,j=1,2$. Again,
it should be noted that $\bold{q'}_{ij}$ for $i,j=1,2$ has one dimension less than $\tilde{\bold{q}}_{ij}$. Hence, (\ref{equ19}) can be rewritten as,
\begin{eqnarray}\label{equ20}
\bold{\Phi}_{21}\bold{y}_{1}&=&d_{11}|\bold{q}_{11}|~|\bold{q'}_{11}|~\tilde{\bold{q'}}_{11}
+d_{12}|\bold{q}_{12}|~|\bold{q'}_{12}|~\tilde{\bold{q'}}_{12}\nonumber\\
&+&d_{21}|\bold{q}_{21}|~|\bold{q'}_{21}|~\tilde{\bold{q'}}_{21}
+d_{22}|\bold{q}_{22}|~|\bold{q'}_{22}|~\tilde{\bold{q'}}_{22}\nonumber\\
&+&\bold{\Phi}_{21}\bold{n}_1,
\end{eqnarray}
where again, it is assumed $\tilde{\bold{q'}}_{ij}=\frac{\bold{q'}_{ij}}{|\bold{q'}_{ij}|}$ is the normalized vector along $\bold{q'}_{ij}$. On the other hand, the normalized vector $\tilde{\bold{q}}_{21}$ can be decomposed as follows,
\begin{eqnarray}\label{equ21}
\tilde{\bold{q}}_{21}=\sqrt{1-a_{1}}\hat{\bold{q}}_{21}+\sqrt{a_{1}}{\bold{q}}^{\perp}_{21}
\end{eqnarray}
where it is assumed $a_{1}\triangleq \sin^{2}(\angle(\tilde{\bold{q}}_{21},\hat{\bold{q}}_{21}))$. Also, ${\bold{q}}^{\perp}_{21}$ is a unit vector in the null space of $\hat{\bold{q}}_{21}$~($N(\hat{\bold{q}}_{21})$).

Therefore, noting $\bold{\Phi}_{21}$ is the projection matrix corresponding to the null space of $\hat{\bold{q}}_{21}$, it follows,
\begin{eqnarray}\label{equ22}
|\bold{q'}_{21}|~\tilde{\bold{q'}}_{21}&=&\bold{\Phi}_{21}\tilde{\bold{q}}_{21}\nonumber\\
&=&\sqrt{1-a_{1}}\bold{\Phi}_{21}\hat{\bold{q}}_{21}+\sqrt{a_{1}}\bold{\Phi}_{21}{\bold{q}}^{\perp}_{21}\nonumber\\
&=&\sqrt{a_{1}}{\bold{q}}^{\perp}_{21}.
\end{eqnarray}

Similarly, noting $\hat{\bold{q}}_{22}=\hat{\bold{q}}_{21}$, we have
\begin{eqnarray}\label{equ23}
|\bold{q'}_{22}|~\tilde{\bold{q'}}_{22}&=&\bold{\Phi}_{21}\tilde{\bold{q}}_{22}\nonumber\\
&=&\sqrt{1-a_{2}}\bold{\Phi}_{21}\hat{\bold{q}}_{22}+\sqrt{a_{2}}\bold{\Phi}_{21}{\bold{q}}^{\perp}_{22}\nonumber\\
&=&\sqrt{a_{2}}{\bold{q}}^{\perp}_{22}~,
\end{eqnarray}
where $a_{2}$ is computed as $a_{2}\triangleq \sin^{2}(\angle(\tilde{\bold{q}}_{22},\hat{\bold{q}}_{21}))$. Moreover, using Lemma~\ref{lemma:upper-bounded}, it follows,
\begin{eqnarray}\label{equ24}
E[a_1]=E[\sin^{2}\big(\angle(\tilde{\bold{q}}_{21},\hat{\bold{q}}_{21})\big)]&<& 2^{-\frac{B}{M-1}},\nonumber\\
E[a_2]=E[\sin^{2}\big(\angle(\tilde{\bold{q}}_{22},\hat{\bold{q}}_{21})\big)]&<& 2^{-\frac{B}{M-1}}.
\end{eqnarray}

Substituting (\ref{equ22}) and (\ref{equ23}) in (\ref{equ20}), it follows,
\begin{eqnarray}\label{equ25}
\bold{\Phi}_{21}\bold{y}_{1}&=&d_{11}|\bold{q}_{11}|~|\bold{q'}_{11}|~\tilde{\bold{q'}}_{11}
+d_{12}|\bold{q}_{12}|~|\bold{q'}_{12}|~\tilde{\bold{q'}}_{12}\nonumber\\
&+&d_{21}~\sqrt{a_{1}}~|\bold{q}_{21}|~{\bold{q}}^{\perp}_{21}
+d_{22}~\sqrt{a_{2}}~|\bold{q}_{22}|~{\bold{q}}^{\perp}_{22}\nonumber\\
&+&\bold{\Phi}_{21}\bold{n}_1.
\end{eqnarray}

Note that the remaining terms $d_{21}~\sqrt{a_{1}}~|\bold{q}_{21}|~\bold{\Phi}_{21}{\bold{q}}^{\perp}_{21}$ and $d_{22}~\sqrt{a_{2}}~|\bold{q}_{22}|~\bold{\Phi}_{21}{\bold{q}}^{\perp}_{21}$ are due to the fact that the interfering vectors are not perfectly aligned. However, noting (\ref{equ24}), they tend to zero as the number of quantization levels increases.

To decode $d_{11}$, the next step is to remove the second term in (\ref{equ25}). To this end, the signal $\bold{\Phi}_{21}\bold{y}_{1}$ in (\ref{equ25}) is projected to the null space of $\tilde{\bold{q'}}_{12}$, thus it should be multiplied again by $\bold{\Phi}_{12}$ defined as follows,
\begin{eqnarray}\label{equ26}
\bold{\Phi}_{12}=I-\tilde{\bold{q'}}_{12}\tilde{\bold{q'}}_{12}^{\dagger}.
\end{eqnarray}

Therefore, noting~(\ref{equ25}), it follows,
\begin{eqnarray}\label{equ27}
\bold{\Phi}_{12}\bold{\Phi}_{21}\bold{y}_{1}&=&d_{11}|\bold{q}_{11}|~|\bold{q'}_{11}|~\bold{\Phi}_{12}\tilde{\bold{q'}}_{11}\nonumber\\
&+&d_{21}\sqrt{a_{1}}~|\bold{q}_{21}|~\bold{\Phi}_{12}~\bold{q}^{\perp}_{21}\nonumber\\
&+&d_{22}\sqrt{a_{2}}~|\bold{q}_{22}|~\bold{\Phi}_{12}~\bold{q}^{\perp}_{22}\nonumber\\
&+&\bold{\Phi}_{12}\bold{\Phi}_{21}\bold{n}_1
\end{eqnarray}

Now, using~(\ref{equ27}), we are ready to compute the received SINR as follows:
\begin{eqnarray}\label{equ31}
SINR=\frac{S}{N+I}
\end{eqnarray}
where $S$ is the intended signal power. Assuming the average transmit power ($p$) is equally distributed between four data streams and considering the beamforming vectors $\bold{v}_{ij}$ for $i,j=1,2$ are of unit norm, thus\footnote{Note that each data stream is sent throughout $M$ time slots, thus the total transmit power during these time slots is $Mp$ which is equality distributed between four data streams.}  $E[|d_{ij}|^2]=\frac{Mp}{4}$ for $i,j=1,2$. Finally, referring to (\ref{equ27}) we have,
\begin{eqnarray}\label{equ32}
S=\frac{Mp}{4}~|\bold{q}_{11}|^{2}~|\bold{\Phi}_{21}\tilde{\bold{q}}_{11}|^{2}~|\bold{\Phi}_{12}\tilde{\bold{q'}}_{11}|^{2}
\end{eqnarray}
Also, the interference power is computed as,
\begin{eqnarray}\label{equ33}
I &=&
\frac{Mp}{4}~|\bold{q}_{21}|^{2}~|\bold{\Phi}_{12}\bold{q}^{\perp}_{21}|^{2}~a_{1}\nonumber\\
&+&\frac{Mp}{4}~|\bold{q}_{22}|^{2}~|\bold{\Phi}_{12}\bold{q}^{\perp}_{22}|^{2}~a_{2}.
\end{eqnarray}

Also, assuming the resulting noise power at the receiver one after applying $\bold{\Phi}_{12}\bold{\Phi}_{21}$ to the received signal vector is one\footnote{Note that we have simply normalized the transmit power to have the noise power with unit variance. On the other hand, this does not change the final result as we are dealing with fairly large values of transmit power.}, the achievable rate due to data stream $d_{11}$ can be computed as,
\begin{eqnarray}\label{equ34}
\hat C=\frac{1}{M}E[\log(1+\frac{S}{I+1})].
\end{eqnarray}
On the other hand, the throughput when the interference terms are thoroughly removed (the interfering vectors are perfectly aligned in the same direction) and assuming the noise power is one, can be computed as,
{\setlength\arraycolsep{1pt}
\begin{eqnarray}\label{equ35}
C&=&\frac{1}{M}E[\log(1+S)]\nonumber\\
&=&\frac{1}{M}E[\log(1+\frac{Mp}{4}~|\bold{q}_{11}|^{2}~|\bold{\Phi}_{21}\tilde{\bold{q}}_{11}|^{2}~|\bold{\Phi}_{12}\tilde{\bold{q'}}_{11}|^{2})]\nonumber\\
\end{eqnarray}}
In what follows, we aim at investigating the number of quantization vectors which guarantees the gap between the aforementioned achievable rates becomes negligible. To this end, we define the following Lemmas,

%

\begin{lemma}\label{lemma:image}
If $\bold{x}$ and $\bold{v}$ are two randomly chosen unit vectors of dimension $M$, then we have $E[|\Phi_{\bold{v}}\bold{x}|^2]=\frac{M-1}{M}$, where $\Phi_{\bold{v}}$ denotes the projection matrix corresponding to the null space of vector $\bold{v}$.
\end{lemma}

\begin{proof}
See Appendix~\ref{sec:appendix3}.
\end{proof}

%

\begin{lemma}\label{lemma:rategap}
The gap between $C$ and $\hat{C}$ is upper bounded as $\frac{1}{M}\log (1+\frac{p}{2}\frac{M(M-2)}{M-1}.2^{\frac{-B}{M-1}})$.
\end{lemma}

\begin{proof}
Subtracting~(\ref{equ34}) from~(\ref{equ35}), it follows,
{\setlength\arraycolsep{1pt}
\begin{eqnarray}\label{equ36}
g(\textrm{SNR}) &=& C-\hat C\nonumber\\
&=&\frac{1}{M}E[\log(1+S)]-\frac{1}{M}E[\log(1+\frac{S}{I+1})]\nonumber\\
&=& \frac{1}{M}E[\log(1+S)]-\frac{1}{M}E[\log(1+S+I)]\nonumber\\
&~&+\frac{1}{M}E[\log(I+1)]\nonumber\\
&\stackrel{(a)}{\leq}& \frac{1}{M}E[\log(1+S)]-\frac{1}{M}E[\log(1+S)]\nonumber\\
&~&+\frac{1}{M}E[\log(I+1)]\nonumber\\
&=&\frac{1}{M}E[\log(1+I)]
\end{eqnarray}}
where $(a)$ is due to the fact that $I$ is a positive quantity and the logarithm function is monotonically increasing function. Also, due to Jensen's inequality, we have
\begin{eqnarray}\label{equ37}
g(\textrm{SNR})\leq \frac{1}{M}E[\log(1+I)] \leq \frac{1}{M}\log(E[1+I]).
\end{eqnarray}
Substituting~(\ref{equ33}) in~(\ref{equ36}), it follows,
{\setlength\arraycolsep{1pt}
\begin{eqnarray}\label{equ38}
g(\textrm{SNR}) &\leq& \frac{1}{M}\log(E[1+I])\nonumber\\
&=& \frac{1}{M}\log (1 \nonumber\\
&+& E[\frac{Mp}{4}~|\bold{q}_{21}|^{2}~|\bold{\Phi}_{12}\bold{q}^{\perp}_{21}|^{2}~a_{1}]\nonumber\\
&+& E[\frac{Mp}{4}~|\bold{q}_{22}|^{2}~|\bold{\Phi}_{12}\bold{q}^{\perp}_{22}|^{2}~a_{2}])
\end{eqnarray}}
One can readily verify that $E[|\bold{q}_{ij}|^{2}]$ for $i,j=1,2$ is equal to $1$. Thus, referring to Lemma ~\ref{lemma:image} for $E[|\bold{\Phi}_{12}\bold{q}^{\perp}_{21}|^{2}]$ and $E[|\bold{\Phi}_{12}\bold{q}^{\perp}_{22}|^{2}]$, We have:
{\setlength\arraycolsep{1pt}
\begin{eqnarray}\label{equ39}
g(\textrm{SNR}) &\leq& \frac{1}{M}\log (1 \nonumber\\
&+&\frac{Mp}{2}\frac{M-2}{M-1}.2^{\frac{-B}{M-1}})\nonumber\\
&=& \frac{1}{M}\log (1+\frac{p}{2}\frac{M(M-2)}{M-1}2^{\frac{-B}{M-1}})
\end{eqnarray}}
\end{proof}
It is worth mentioning that $\bold{q}^{\perp}_{21}$ and $\bold{q}^{\perp}_{22}$ are of dimension $M-1$, thus $E[|\bold{\Phi}_{12}\bold{q}^{\perp}_{21}|^{2}]=E[|\bold{\Phi}_{12}\bold{q}^{\perp}_{22}|^{2}]=\frac{M-2}{M-1}$.

Finally, to have the rate gap approaching zero and noting (\ref{equ39}), it follows $B=(M-1)\log(p)+\theta(\log\log(p))$ and noting $M=3$ for two-user X channel, thus $B=2\log(p)+\theta(\log\log(p))$, where according to the knuth's notation, $f(n)=\theta{(g(n))}$ means $\textrm{lim}|\frac{f(n)}{g(n)}|=c\leq \infty$ as $n$ tends to infinity\cite{knuth}.

\section{The Impact of Low Feedback rate}\label{sec:Quantrate}
The preceding section yields the appropriate feedback rate, i.e. $(M-1)\log(p)+\theta(\log\log(p))$, which is employed so as to do not have multiplexing gain lost. On the other hand by decreasing quantization rate, the magnitude of quantization error becomes larger and can not be sure that there is no multiplexing gain lost. In order to determine such multiplexing lost, this section aims to investigate the impact of low feedback rate on the multiplexing gain in two user~X channel. The remains of this section follows by analysis of rise and reduction in the amount of quantization rate for large enough $\textrm{SNR}$ which indicate the power constraint.

To scale \emph{B} at rate larger than $(M-1)\log(p)$, i.e. $(M-1)\log(p)+\theta(\log\log(p))$, The rate gap , i.e. $g(\textrm{SNR})$, tends to zero
{\setlength\arraycolsep{1pt}
\begin{eqnarray}\label{equ42}
\lim_{p \rightarrow \infty} g(\textrm{SNR}) \leq \lim_{p \rightarrow \infty}\frac{1}{M}\log (1+\frac{p}{2}\frac{M(M-2)}{M-1}2^{\frac{-B}{M-1}}))=0\nonumber\\
\end{eqnarray}}
Thus the throughput of the proposed method converges to the the throughput when the interference signals are perfectly aligned in the same direction. In the other words, by quantization rate which is grater than \emph{B}, the magnitude of quantization error is venial.

If \emph{B} is scaled at a rate lower than $(M-1)\log(p)$, i.e. $B=\alpha \log(p)$ and $\alpha \leq (M-1)$ , the throughput achieves a multiplexing gain of $\frac{4}{3}(\frac{\alpha}{M-1})$ in two user X channel.
To prove this, the remains of this section approaches the degradation of quantization rate for one message intuitively and the result is scaled finally by $\frac{4}{3}$ which is the achieved \emph{DoF} of the proposed method for two user X channel. For more details see \emph{Theorem 4} of \cite{Jindal-finite_feedback}.

Therefore in the infancy of attaining of the \emph{DoF} of proposed method for degraded quantization rate, received signal to noise-interference ratio, i.e. SINR, in order to evaluate throughput should be extracted. As was obtain in equations~\ref{equ32} and~\ref{equ33} the signal and interference power is proportional with $p$ and the product of $p$ and quantization error respectively. Since the expectation of quantization error is upper bounded by $2^{-\frac{B}{M-1}}$ and \emph{B} is scaled by $\alpha \log(p)$ supposedly, then the interference power scales as $p^{(1-\frac{\alpha}{M-1})}$. So SINR is proportional with $p^{\frac{\alpha}{M-1}}$ and throughput can be obtained.
\begin{eqnarray}\label{equ43}
\hat C&=&E[\log(1+SINR)]\nonumber\\
&=&E[\log(1+p^{\frac{\alpha}{M-1}})]
\end{eqnarray}
and yields the \emph{DoF} as
\begin{eqnarray}\label{equ7}
DoF&=&\lim_{p \to \infty} \frac{\hat C(p)}{\log(p)}\nonumber\\
&=&\frac{\alpha}{M-1}
\end{eqnarray}
where \emph{DoF} is for one message. Since the propose method can achieve $\frac{4}{3}$ multiplexing gain thus the total \emph{DoF} can be achieved is $\frac{4}{3}(\frac{\alpha}{M-1})$.


\section{Conclusion}
This paper aims to investigate the required feedback rate  to achieve the Degrees of Freedom~(~\emph{DoF}) of the generalized ergodic X channel. To this end, through using the notion of interference alignment technique, a fixed precoding approach is proposed which attempts to align the quantized version of non-intended signals~(interferers) in the same direction at the receiver part. Finally, it is shown a feedback rate of $B=2\log(p)+\theta(\log\log(p))$ is adequate to achieve the~\emph{DoF}.

\section{Acknowledgment}\label{sec:acknowledgment}
The authors would like to greatly knowledge the financial support of Iran Telecommunication Research Center~(ITRC) under grant number 3133/500.

\appendices
\section{Proof of lemma~\ref{lemma:upper-bounded}}\label{sec:appendix1}
In this appendix, we prove the expectation  quantization error can be upper bounded by $2^{-\frac{B}{M-1}}$.
{\setlength\arraycolsep{2pt}
\begin{eqnarray*}
E[\sin^{2}(\angle(\tilde{\bold{q}}_{i},\hat{\bold{q}}_{i})]&\stackrel{a}{=}&\int_{0}^{1}((1-x^{M-1})^{2^{B}})dx\\
&\stackrel{b}{=}&\frac{1}{M-1}\beta(2^{B}+1,\frac{1}{M-1})\nonumber\\
&\stackrel{c}{=}&\frac{\frac{1}{M-1}\Gamma(2^{B}+1)\Gamma(\frac{1}{M-1})}{\Gamma(2^{B}+1+\frac{1}{M-1})}\\
&\stackrel{d}{=}&\frac{2^{B}\Gamma(2^{B})\Gamma(1+\frac{1}{M-1})}{\Gamma(2^{B}+1+\frac{1}{M-1})}\\
&\stackrel{e}{\leq}& \frac{2^{B}\Gamma(2^{B})}{\Gamma(2^{B}+1+\frac{1}{M-1})}\\
&\stackrel{f}{=}&\frac{\Gamma(2^{B}+1)}{\Gamma(2^{B}+1+\frac{1}{M-1})}\\
&\stackrel{g}{<}& (2^{B}+\frac{M}{2(M-1)})^{-\frac{1}{M-1}}\\
\end{eqnarray*}}
Where $a$ is hold on with due attention to get integrations by part from $E(x)= \int_{0}^{1} xf(x) dx$ and $b$ is followed by the integral representation for the beta function is given in~\cite{GuptaHandbook,Jindal-finite_feedback}:
\begin{displaymath}
\beta \left( c,\frac{a}{b}\right)=b\int_{0}^{1} x^{a-1}\left(1-x^{b}\right)^{c-1}dx,\qquad a>0, b>0, c>0.
\end{displaymath}

The equality $c$, $d$ and $f$ have been given from of definition of beta and gamma function, sequentially. The inequality $e$ is based on that for $1\leq x\leq2$, the gamma function is not larger than one~\cite{Davis-gamma_function,Jindal-finite_feedback}. The last inequality is reached by implementing Kershaw's inequality for the gamma function~\cite{Kershaw-gamma_function,Jindal-finite_feedback}:
\begin{displaymath}
\frac{\Gamma\left(x+s\right)}{\Gamma\left(x+1\right)}< \left(x+\frac{s}{2}\right)^{s-1},\qquad \forall x>0, 0<s<1.
\end{displaymath}
In the above inequality if $x$ and $s$ are replaced by $2^{B}+\frac{1}{M-1}$ and $1-\frac{1}{M-1}$, then we have:
\begin{displaymath}
\frac{\Gamma\left(2^{B}+1\right)}{\Gamma\left(2^{B}+1+\frac{1}{M-1}\right)}< \left(2^{B}+\frac{M}{2\left(M-1\right)}\right)^{-\frac{1}{M-1}},
\end{displaymath}
And due to the fact that the function $\left(.\right)^{-\frac{1}{M-1}}$ is decreasing,we have:

\begin{displaymath}
E[\sin^{2}(\angle(\tilde{\bold{q}}_{kl},\hat{\bold{q}}_{kl})]<2^{-\frac{B}{M-1}}.
\end{displaymath}

\section{Proof of lemma~\ref{lemma:image}}\label{sec:appendix3}
In this appendix, we want to show the expected value of the squared magnitude of the projection of $M$ dimensional unit norm vector $\bold{x}$ on the null space of another isotropically distributed vector $\bold{v}$ is equal to $\frac{M-1}{M}$. Assuming $\tilde{\bold{v}}=\frac{\bold{v}}{|\bold{v}|}$ and noting $\bold{\Phi}_{\bold v}=\bold{\Phi}_{\tilde{\bold v}}=\bold{I}-\tilde{\bold v}\tilde{\bold v}^{\dagger}$, it follows,
\begin{figure}[htp]
\centering
\epsfig{file=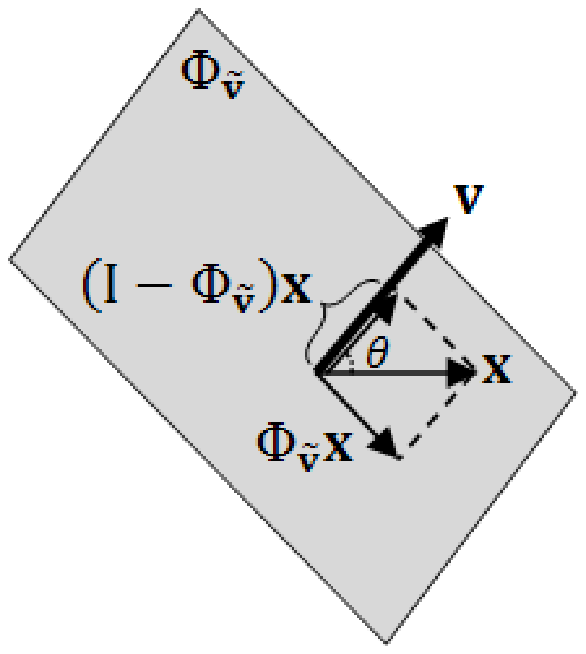,width=0.5\linewidth,clip=}
\caption{The projection of $M$ dimensional vector $\bold{x}$ on the null space of vector $\bold{v}$.}
\label{fig:3}
\end{figure}

\begin{eqnarray}\label{equa1}
\bold{\Phi}_{\tilde{\bold v}}\bold x = \bold x - \tilde{\bold v}(\tilde{\bold v}^{\dagger}.\bold x).
\end{eqnarray}
And due to inner product property, we have
\begin{eqnarray*}
\cos(\theta)= \tilde{\bold v}^{\dagger}.\bold x~,
\end{eqnarray*}
where it is assumed the angle between vectors $\bold{x}$ and $\bold{v}$ is $\theta$ (see Figure~\ref{fig:3}). Since $\bold{x}$ and $\bold{v}$ are isotropically distributed then the distribution of $|\cos(\theta)|^{2}$ is beta with parameters $1$ and $M-1$ parameters~\cite{Jindal-finite_feedback}. Thus, noting (\ref{equa1}) and assuming $\bold{x}$ is of unit norm, we have,
\begin{eqnarray*}
|\bold{\Phi}_{\tilde{\bold v}}\bold x|^2 &=& \bold x^{\dagger}\bold x+(\bold x^{\dagger}\tilde{\bold v})\tilde{\bold v}^{\dagger}\tilde{\bold v}(\tilde{\bold v}^{\dagger}\bold x)-2\bold x^{\dagger}\tilde{\bold v}\tilde{\bold v}^{\dagger}\bold x\\
&=& 1-|\cos(\theta)|^{2}
\end{eqnarray*}

Finally, we have,
\begin{eqnarray*}
E[|\bold{\Phi}_{\tilde{\bold v}}\bold x|^2] = 1 - E[|\cos(\theta)|^{2}] \stackrel{(a)}{=} \frac{M-1}{M}.
\end{eqnarray*}
which $(a)$ is due to the fact that the expectation of a beta distribution
with parameters $\alpha$ and $\beta$ is equal to $\frac{\alpha}{\alpha+\beta}$~\cite{GuptaHandbook}. In our case, we have $\alpha=1$ and $\beta=M-1$.
\bibliographystyle{IEEEtran}
\bibliography{paper}

\end{document}